\documentclass[journal,onecolumn,draftcls]{IEEEtran}

\usepackage[usenames]{color} 
\usepackage{amsmath} 
\usepackage{amsfonts} 
\usepackage{array} 
\usepackage[english]{babel} 
\usepackage{amsthm} 
\usepackage{color}
\usepackage{graphicx} 

\usepackage{pgf}
\usepackage{tikz}
\usetikzlibrary{arrows,automata}
\usetikzlibrary{positioning}
\usetikzlibrary{backgrounds}
\usetikzlibrary{matrix}
\usetikzlibrary{decorations.pathmorphing}
\usetikzlibrary{fit}

\newcommand{\F}{{\mathbb{F}}}

\newtheorem{lemm}{Lemma}
\newtheorem{theo}[lemm]{Theorem}
\newtheorem{prop}[lemm]{Proposition}

\newtheorem{coro}[lemm]{Corollary}

\title{Quasi-cyclic Flexible Regenerating Codes}
\author{Bernat Gast\'on, Jaume Pujol, and Merc\`e
Villanueva \\
\small Department of Information and Communications Engineering,
\small Universitat Aut\`onoma de Barcelona \\
\small \texttt{\{Bernat.Gaston $|$ Jaume.Pujol $|$
Merce.Villanueva\}@uab.cat} \\
}
\begin{document}

\maketitle

\begin{abstract}
In a distributed storage environment, where the data is placed in nodes connected
through a network, it is likely that one of these nodes fails. It is known
that the use of erasure coding improves the fault tolerance and minimizes the
redundancy added in distributed storage environments. The use of regenerating
codes not only make the most of the erasure coding improvements, but
also minimizes the amount of data needed to regenerate a failed node.

In this paper, a new family of regenerating codes based on quasi-cyclic codes
is presented. Quasi-cyclic flexible minimum storage regenerating (QCFMSR)
codes are constructed and their existence is proved.
Quasi-cyclic flexible regenerating codes with minimum bandwidth constructed from
a base QCFMSR code are also provided. These codes not only achieve optimal MBR
parameters in terms of stored data and repair bandwidth, but also for an
specific choice of the parameters involved, they can be decreased under the
optimal MBR point.

Quasi-cyclic flexible regenerating codes are very interesting because of
their simplicity and low complexity. They allow exact repair-by-transfer in the
minimum bandwidth case and an exact pseudo repair-by-transfer in the MSR case,
where operations are needed only when a new node enters into the system
replacing a lost one.

\end{abstract}

\section{Introduction}
\label{sec:1}

The availability problem of stored data is an essential issue, which has
been studied intensively lately \cite{Pl01}. The increasing use of
distributed storage systems (DSS), like cloud storage, to a massive scale has
changed the paradigm of data storage. The amount of stored data in a DSS must be
minimized and the problem of storage device failures must be addressed.

The motivation for using erasure coding in DSS comes from the need to keep the
information available even when storage device failures occur and,
at the same time, to reduce the stored amount of data. Erasure codes in DSS
allow to achieve high fault tolerance (resistance to storage device failures
without losing information), requiring less storage
overhead (redundancy added to repair the files) that the one required by a data
replication scheme \cite{We01},
also known as backups. There are the obvious cost savings from purchasing less
hardware to store the data, but there are also
significant savings from the fact that this also reduces data centers
size, the power for running less hardware, among with other savings \cite{Hu01}.

The use of erasure coding in DSS minimizes the amount of
data stored in the system, but introduces what is known as the code repair
problem \cite{Pa01}: how to maintain the encoded representation given by the
erasure code when storage device failures occur, minimizing the amount of
stored data per node $\alpha$ and the bandwidth used to regenerate one node
$\gamma$. To maintain the same
redundancy when a storage node (device) leaves the system, a newcomer has
to join the system, access some existing nodes, download data from them, and
replace the lost node. This operation is not an exception, in fact, it is
usual.

This paper is organized as follows. In Section \ref{sec:2}, the repair problem
is exposed and some of the most known techniques used to minimize it are shown.
In Section \ref{MSRcodes}, the quasi-cyclic flexible minimum storage
regenerating codes are constructed, and their bounds and properties are proved.
In Section \ref{MBRcodes}, a technique used to produce minimum bandwidth codes
from minimum storage codes is formally explained and it is applied to the
proposed minimum storage codes to produce quasi-cyclic flexible regenerating
codes with minimum bandwidth. These codes are compared with other
regenerating codes with the same parameters. Finally, in Section
\ref{Conclusions}, the conclusions of this work are exposed.

\section{Minimizing the repair problem}
\label{sec:2}

There are various techniques used to minimize the bandwidth needed to
regenerate a failed node. In this section, we define the problem, and we review
some of the most well known techniques used in DSS to solve this problem.

Let $M$ be the size of a file. In a replication scheme, if each node
stores $\alpha$ data units, with $\alpha \ll M$, the newcomer must download
$\alpha$ data units to replace the lost node. However, in an erasure coding
scheme, the same newcomer must download $M$ data units to store only $\alpha$.

Denote $\F_{q}$ the finite field of $q$ elements. A linear code is a
linear map from $\F_{q}^k$ to $\F_{q}^n$ which converts vectors
of $k$ coordinates over $\F_{q}$ to vectors of $n$ coordinates over the same
field  called codewords. Note that $k$ is the dimension and $n$ the length of
the linear code, so the rate of the code is $R=k/n$. Let $d$ be the minimum
distance of the linear code. It is
well known that $d \le n-k+1$ \cite{MW}. When one or more coordinates of a
codeword are lost, it is possible to obtain the original vector over $\F_{q}^k$
by using an algorithm that requires at least $n-d+1$ correct coordinates of the
codeword, this
procedure is called decoding. We can say that a $[n,k,d]$ linear code encodes
$k$ symbols of $\F_q$ into $n$ symbols of $\F_q$ and repairs $d-1$ erasures.
From now on, we assume that the codes used are linear and are designed to repair
erasures. The classical decoding using linear codes always needs and uses
$n-d+1$ correct coordinates no matter if the codeword has one or more (up to
$d-1$)
erasures. Now, the next question arises, is it possible to repair one single
erasure requiring less than $n-d+1$ coordinates?

Take a file of size $M$ and split it into $k$ pieces of size $M/k$ over $\F_{q}$
organized as
a vector $v=(v_1,\ldots,v_k)$. Encode this vector using a $[n,k,d]$ code which
produces a codeword $c=(c_1,\ldots,c_n)$. Assume that each one of the
coordinates $c_i$ is stored in a different storage node $s_i$,
$i=1,\ldots,n,$ of a DSS. Note that any subset of $n-d+1$ of these nodes is
enough to recover the file. Moreover, the fault tolerance of the DSS is $d-1$,
since it is the maximum number of node failures that can be tolerated by the
DSS, where tolerated means that the DSS is still able to recover any piece of
the original stored information. If the code used is a Maximum Distance
Separable (MDS) code \cite{MW}, which is considered the best code in terms of
storage efficiency, the amount of redundancy in the system is minimum for a
given fault tolerance and $d=n-k+1$, so the file can be recovered by connecting
and downloading the data from any $k$ storage nodes.

The goal is to minimize the amount of stored data per node $\alpha$ and the
repair bandwidth $\gamma$, that is, the amount of data required to regenerate a
node. These parameters $\alpha$ and $\gamma$ can be
reduced by using at least two different techniques, which will be summarized in
the next two subsections.

\subsection{Locally Repairable Codes}

Let $C$ be an $[n,k,d]$ linear code over $\mathbb{F}_q$. We say that a
coordinate $i$ of $C$ has repair degree $r_i$ if we can recover any symbol at
coordinate $i$ by accessing at least $r_i$ other codewords symbols. The repair
degree of $C$, $r$, is the maximum of $r_i$, $i=1,\ldots,n$, and $C$ is called
an $r$ locally reparable code ($r$-LRC). Note that in a $r$-LRC any message
symbol can be recovered by accessing at most $r$ other codeword symbols.

In \cite{Go01}, it is shown that the minimum distance $d$ of a code is
uper bounded by $d \le n-k-\lceil \frac{k}{r} \rceil +2$, which means that a
as $r$ increases (approaching to $k$) $d$ decreases. It is also known that it is
a good metric for repair cost \cite{Og01}, \cite{Pa02}. Moreover, note that the
MDS codes have degree $r=k$. Locally repairable codes (LRC) try to keep $r$ at
very low rates, this means that for LRC it is a goal that $r \ll k$. Moreover,
the LRC codes are not MDS.

\subsection{Regenerating Codes}
\label{subsec:2.1}

Let $\F_q$ be a base field with $q$ elements and let $\F_{q^t}$ be an extended
field of $\F_q$ with $q^t$ elements. This means that each element of $\F_{q^t}$
is composed by $t$ elements of $\F_q$ called coordinates. Using an $[n,k,d]$
code over $\F_{q^t}$, we can encode an information
vector $v \in \F_{q^t}^k$ and produce a codeword $c \in \F_{q^t}^n$. This kind
of codes are called array codes, and each coordinate of $v$ or $c$ is called an
array coordinate. A well known example of an array code is the EVENODD code
\cite{Th01}.


In \cite{Di01}, a solution to minimize the bandwidth used to regenerate a
failed node is proposed. This solution consists of using regenerating codes. The
main idea of this kind of codes is the use of a network multicast technique for
data transmission optimization, called network coding \cite{NC01}, in
conjunction with array codes. The regenerating codes use linear codes seen over
the extended field to minimize the amount of stored data and use
network coding and the codes seen over the base field, to minimize the
bandwidth used to regenerate a failed node.

Let $C$ be a $\left[n,k,r\right]$ regenerating code, where the length $n$
is the total number of nodes in the system; the dimension $k$ is the value such
that any $k$ nodes contain the minimum amount of information necessary to
reconstruct the file; and the cardinal of the
set of helper nodes $r$ is the number of nodes necessary to regenerate
one failed node. Let $v \in \F_{q^t}^k$ be the information vector
representing the file to be stored in a DSS. Let $c \in \F_{q^t}^n$ be the
corresponding codeword after encoding $v$ using $C$. Note that if each array
coordinate of $c$ is stored in a node, and the code used is MDS, the parameters
$n$, $k$ and $d$ of the regenerating code coincide with the parameters $n$,
$k$, and $d$, of the linear code.
Moreover, in this case, $C$ is called a Minimum Storage Regenerating (MSR) code,
since it minimizes the storage overhead and $d=n-k+1$. However, it is also
possible to store more data in the same storage node, for example by adding an
extra set of elements over the base field, producing more redundancy, which
could be used to regenerate a failed node requiring less repair bandwidth. Using
this technique, it is possible to minimize the repair problem at the cost of
some extra overhead but maintaining $d=n-k+1$. When $C$ achieves the minimum
$\alpha$ such that $\alpha=\gamma$, $C$ is called a Minimum Bandwidth
Regenerating (MBR) code, and it can be seen that the parameters $n$, $k$ and $d$
do not coincide with the ones defined for a linear code.

Regenerating codes assume the data reconstruction condition: any $k$ nodes must
be enough to recover the file, which means that the minimum
distance must be $d=n-k+1$, so it is necessary to have $\alpha k \ge M$.
Moreover, if $d=n-k+1$ and $\alpha k = M$, we have an MDS code. Another
condition is
that the regeneration of any node in the system must require less
repair bandwidth than the total file size $M$, that is $\gamma < M$. If each
helper node sends $\beta$ data units, the repair bandwidth used is
$\gamma= \beta r$.  Then, $r$ must achieve $k < r < n$, unlike in LRC, since the
set of possible newcomers can not be greater than $n-1$, and if $r=k$, there is
no possible optimization in the repair bandwidth. As we will see, in order to
decrease the repair bandwidth by connecting to $r$ nodes, each helper node must
send less than $\alpha$ data units, so they send only a linear combination of
their coordinates over the base field.

Note the difference between LRC and regenerating codes. In both cases, the
repair degree is the number of helper nodes necessary to regenerate a failed
node. However, in LRC, each coordinate of a codeword $c$ is stored in one node,
and when we access to one helper node it means that we download the entire
coordinate contained in it. Therefore, in order to decrease $\gamma$ we need $r
\ll k$. In regenerating codes, to maintain the distance $d=n-k+1$, we need $r
\ge k$. Thus, $c$ is composed by array coordinates, the solution is to
increase $r$ but downloading less than an entire array coordinate from each
helper node. In other words, LRC are the result of creating new codes adapted to
a distributed storage systems, while regenerating codes are classical codes in
which a network coding technique is used to reduce the repair bandwidth.

In \cite{Di01}, the authors modelize the life of a DSS with a graph called
\textit{information flow graph}. This graph has three types of weighted edges:
internal storage edges that represent the amount of data stored per
node $\alpha$; regenerating edges that represent the amount of data sent by each
one of the helper nodes $\beta$; and edges that connect the
Source (the file to be stored) and the Data Collector (the user who access the
file) with the storage nodes. From this graph, one can compute the mincut
between the Source ($S$) and the Data Collector ($DC$) denoted by
$\mbox{mincut}(S,DC)$ and, after an optimization process, minimize $\alpha$ and
$\gamma$ resulting in a threshold function.

Despite it is not the aim of this paper to focus on information flow graphs,
Figure \ref{Fig:1} illustrates one of these information flow graphs for a
$[4,2,3]$ regenerating code. In general, the mincut equation is
$\mbox{mincut}(S,DC) \ge
\sum_{i=0}^{k-1} \min ( (r-i)\beta,
\alpha ) \ge M$, which means that for each newcomer up to $k$, it is composed by
either the weight of the set of regenerating edges or
the weight of the internal storage edge of the newcomer. Then, the optimization
over $\alpha$ and $\gamma$ gives the function

\begin{equation}
\label{eq:1}
 \alpha^*(r,\gamma) = \left\{ \begin{array}{lc}
             \frac{M}{k}, & \gamma \in [f(0), +\infty) \\
             \\ \frac{M-g(i)\gamma}{k-i}, & \gamma \in [f(i),f(i-1)) \\
                & i=1, \ldots, k-1,
             \end{array}
   \right.
\end{equation}
where
$$f(i) = \frac{2Mr}{(2k-i-1)i + 2k(r-k+1)} \text{ and }$$
$$g(i) = \frac{(2r-2k+i+1)i}{2r}.$$

\begin{figure}
\centering
\begin{tikzpicture}[shorten >=1pt,->]
  \tikzstyle{vertix}=[circle,fill=black!25,minimum size=18pt,inner sep=0pt,
node distance = 0.8cm,font=\tiny]
  \tikzstyle{invi}=[circle]
  \tikzstyle{background}=[rectangle, fill=gray!10, inner sep=0.2cm,rounded
corners=5mm]

  \node[vertix] (s) {$S$};

  \node[vertix, right=1cm of s] (vin_2)  {$v_{in}^2$};
  \node[vertix, below of=vin_2] (vin_3)  {$v_{in}^3$};
  \node[vertix, above of=vin_2] (vin_1)  {$v_{in}^1$};
  \node[vertix, below of=vin_3] (vin_4)  {$v_{in}^4$};

 \foreach \to in {1,2}
   {\path (s) edge[bend left=20,font=\tiny] node[anchor=south,above]{$\infty$}
(vin_\to);}
 \foreach \to in {3,4}
   {\path (s) edge[bend right=20,font=\tiny]  node[anchor=south,above]{$\infty$}
 (vin_\to);}

  \node[vertix, right of=vin_1] (vout_1)  {$v_{out}^1$};
  \node[vertix, right of=vin_2] (vout_2)  {$v_{out}^2$};
  \node[vertix, right of=vin_3] (vout_3)  {$v_{out}^3$};
  \node[vertix, right of=vin_4] (vout_4)  {$v_{out}^4$};

  \node[vertix, right of=vout_4, node distance = 2cm] (vin_5)
{\scriptsize{$v_{in}^5$}};
  \node[vertix, right of=vin_5] (vout_5)
{\scriptsize{$v_{out}^5$}};

  \node[vertix, right of=vout_1, above of=vin_5, node distance = 2cm] (vin_6)
{\scriptsize{$v_{in}^6$}};
  \node[vertix, right of=vin_6] (vout_6)
{\scriptsize{$v_{out}^6$}};

  \node[vertix, right of=vout_6, node distance = 1.5cm] (DC)
{DC};

   \path (vout_5) edge[bend right=20,font=\tiny]
node[anchor=south,above]{$\infty$}
 (DC);
   \path (vout_6) edge[bend left=10,font=\tiny]
node[anchor=south,above]{$\infty$}
 (DC);

  \path[->, bend left=15,font=\tiny] (vout_2) edge
node[anchor=south,above]{$\beta$} (vin_5);
  \path[->, bend left=10,font=\tiny] (vout_3) edge
node[anchor=south,above]{$\beta$}(vin_5);
  \path[->, bend left=20,font=\tiny] (vout_1) edge
node[anchor=south,above]{$\beta$}(vin_5);

  \path[->, bend left=15,font=\tiny] (vout_1) edge
node[anchor=south,above]{$\beta$} (vin_6);
  \path[->, bend left=10,font=\tiny] (vout_2) edge
node[anchor=south,above]{$\beta$}(vin_6);
  \path[->, bend left=10,font=\tiny] (vout_5) edge
node[anchor=south,above]{$\beta$}(vin_6);

 \foreach \from/\to in {1,2,3,4,5,6}
  { \path[->,font=\tiny] (vin_\from) edge node[anchor=south] {$\alpha$}
(vout_\to); }

  \draw[-,color=red,thick] (1.2,-0.5) -- (2.7,-1.2);
  \draw[-,color=red,thick] (1.2,-1.2) -- (2.7,-0.5);

  \draw[-,color=red,thick] (1.2,-1.2) -- (2.7,-1.9);
  \draw[-,color=red,thick] (1.2,-1.9) -- (2.7,-1.2);
\end{tikzpicture}

\caption{Information flow graph corresponding to a $[4,2,3]$ regenerating
code.}
\label{Fig:1}
\end{figure}
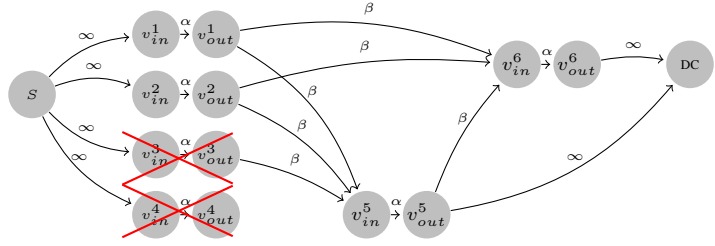

For an specific set of parameters $n$, $k$ and $r$ in function (\ref{eq:1}), it
is possible to find the optimal tradeoff curve representing the minimization of
$\alpha$ and
$\gamma$ as shown in Figure \ref{Fig:2} for a $[10,5,9]$ regenerating code. The
optimal parameters of the MSR code are given by the point with the minimum
$\alpha$, while the optimal parameters of the MBR code are given by the point
with the minimum $\gamma$. Between these two extremal points, there exist some
internal points which represent the  limits of the intervals that appear in the
threshold function (\ref{eq:1}).

\begin{figure}
\centering
\includegraphics[width=9cm, height=5cm]{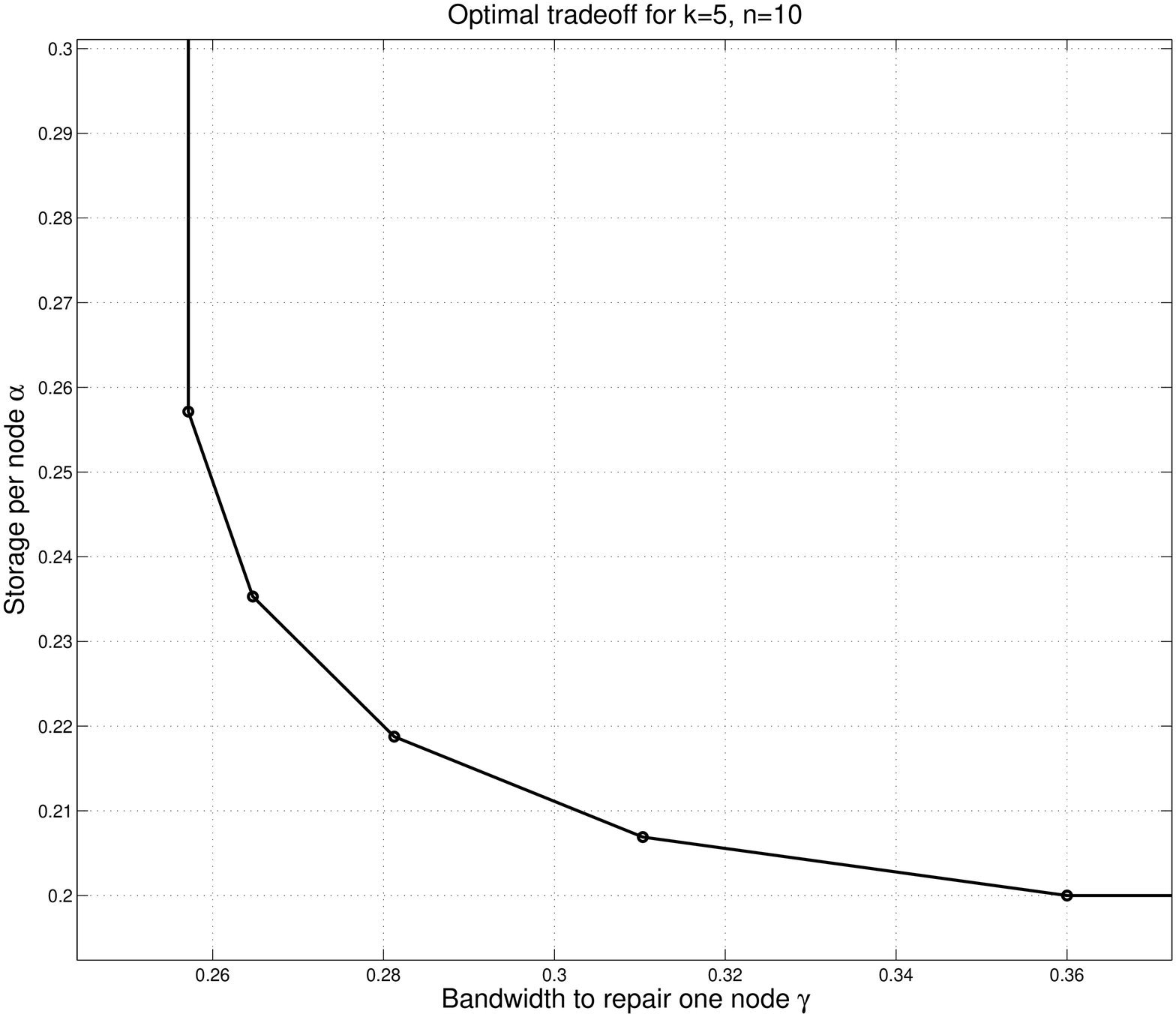}
\caption{
Optimal tradeoff curve between $\alpha$ and $\gamma$ for a $[10,5,9]$
regenerating code.
}
\label{Fig:2}
\end{figure}

To achieve this tradeoff curve, multiple techniques have been used, like
interference alignment \cite{Shah02}, product-matrix construction \cite{Ku02},
or designs \cite{Ro01}, among others. Behind these techniques, there are two
main ideas: the use of array codes, which allows the DSS to treat the data
inside a node as a set
of small independent coordinates over the base field; and the use of network
coding to send linear combinations of these coordinates through the network.

If any newcomer is able to exactly replicate the lost node, we say that the DSS
has the \textit{exact repair} property. Otherwise, if the newcomers store a
linear combination that does not reduce the dimension of $C$ but it is not
exactly the same as the included in the lost node, we say that the DSS has the
\textit{functional repair} property \cite{Di05}. Exact repair is much more
desirable than functional repair, since despite the number of failed nodes that
the DSS has repaired over an interval of time, it is possible to use systematic
encoding of the information and keep this systematic representation over the
time.
This means that there is always one accessible copy of the original file stored
in the DSS. It is worth to mention that in 
\cite{Ku03} it is proved that the interior points of the tradeoff curve are not
achievable using exact repair.

We say that a DSS has the \textit{uncoded repair} property if it is possible to
replace a failed node without doing any linear operation in the newcomer
neither in the helper nodes. There exist uncoded constructions for the MBR point
like the ones shown in \cite{Ku02} and \cite{Ro01}. However, for the MSR point,
there only exist uncoded constructions using functional repair \cite{Hu03}.

\section{Quasi-cyclic Flexible MSR codes}
\label{MSRcodes}

In this section, we describe the quasi-cyclic flexible minimum storage
regenerating (QCFMSR) codes in detail. Specifically, in Subsection
\ref{CodeConstruction}, we show how to construct them and some of their
properties; in Subsection \ref{Regeneration}, we see how to regenerate a
failed node; in Subsection \ref{DataReconstruction}, we prove their
existence by showing that the data reconstruction condition is achieved; and
finally, in  Subsection \ref{ExampleMSR}, we describe an example of a
$[6,3,4]$ QCFMSR code.

\subsection{Code Construction}
\label{CodeConstruction}

Let $C$ be an array code of length $n=2k$ and dimension $k$
over $\F_{q^2}$ constructed from the nonzero coefficients
$\zeta_1,\ldots,\zeta_k$ over $\F_q$, and for which the encoding is done over the base
field $\F_q$ in the following way. An information vector $v\in \F^k_{q^2}$ is seen as a vector $v=(v_1,\ldots,v_{n})$ over $\F_q$,
and is encoded into $c\in \F^{n}_{q^2}$ seen as a codeword $c=(c_1,\ldots,c_{2n})=(v_1,\ldots,v_{n},\rho_1,\ldots,\rho_{n})$ over $\F_q$,
where the redundancy coordinates $\rho_1,\ldots,\rho_{n}$ are given by the
following equation:
\begin{equation}
\label{eq:2}
 \rho_i = \sum_{j=i+1}^{k+i} \zeta_{j-i} v_{j} \quad i=1,\ldots,n,
\end{equation}
where $\zeta_l \in \F_q \setminus \left\lbrace 0 \right\rbrace$ for
$l=1,\ldots,k$ and $j=i+1,\ldots,k+i \mod n$. The rate of the code is $R=1/2$ and the encoding over $\F_q$ is done by using a
quasi-cyclic code \cite{MW} as we will see later. Quasi-cyclic codes are known
by their simplicity for encoding-decoding operations.

A $[2k,k,r]$ QCFMSR code over $\F_{q^2}$ is a regenerating code constructed from the array code $C$.
Take a file of size $M$ and split it into $k$ pieces over $\F_{q^2}$, or equivalently, into $n=2k$ pieces over $\F_q$ organized as
a vector $v=(v_1,\ldots,v_{n})$ over $\F_q$.
The $[2k,k,r]$ QCFMSR code over $\F_{q^2}$ is composed by a set of $n=2k$ storage nodes, denoted by $\left\lbrace s_1, s_2,\ldots, s_n \right\rbrace $, where each storage node $s_i$, $i=1,\ldots,n$, stores two coordinates over $\F_q$, $(v_i,\rho_i)$, which can be seen as one array coordinate over
$\F_{q^2}$. The size of each coordinate over $\F_q$ is $M/2k$ and the size of each array coordinate stored in $s_i$ is
$\alpha = M/k$.

%
%
%
%

Let $S$ be the set of all subsets of $\{1,\ldots,n\}$ of size $k$. Let $D$ be an
$n\times n$ matrix over $\F_q$ and let $s=\{i_1,\ldots,i_k\}\in S$.  Let
$D^{\left\lbrace i \right\rbrace }$ denote the $i$th column vector of $D$ and
$D^s$ denote the $n\times k$ submatrix of $D$ given by the $k$ columns
determined by the set $s$.

Let $F = \left( I|Z \right) $ be a $n \times 2n$ matrix, where $I$ is the
$n \times n$ identity matrix, and $Z$ is a $n \times n$ circulant matrix
defined from the nonzero coefficients $\zeta_1,\ldots,\zeta_k$ as follows:
\begin{equation}
  Z=\left ( \begin{array}{cccccccc}
  0 & 0 & \cdots & 0 & \zeta_k & \zeta_{k-1} & \cdots & \zeta_1 \\
  \zeta_1 & 0 & \cdots & 0 & 0   & \zeta_k   & \cdots & \zeta_2 \\
  \vdots & \vdots & \vdots & \vdots & \vdots & \vdots & \vdots & \vdots \\
  \zeta_k & \zeta_{k-1} & \cdots & \zeta_1 & 0 & 0 & \cdots & 0 \\
  0   & \zeta_k & \cdots & \zeta_2 & \zeta_1 & 0 & \cdots & 0 \\
   \vdots & \vdots & \vdots & \vdots & \vdots & \vdots & \vdots & \vdots \\
   0 & 0 & \cdots & \zeta_k & \zeta_{k-1} & \zeta_{k-2} & \cdots & 0 \\
  \end{array}\right ).
\end{equation}
The matrix $F$ represents the array code $C$, so also the QCFMSR code constructed from $C$.
Each row is the encoding of one
coordinate over the base field $\F_q$, and each node is represented by two
columns, one from $I$ and another one from $Z$. Actually, the node $s_i$, which
stores $(v_{i},\rho_i)$, is also given by
$$
 \left( v_i, \rho_i \right)  = (vI^{\left\lbrace i \right\rbrace
},vZ^{\left\lbrace i \right\rbrace }).
$$
Note that the information coordinates are represented by the identity matrix
$I$, while the redundancy coordinates are represented by the circulant matrix
$Z$.

Circulant matrices have been deeply studied because of their symmetric properties
\cite{Circ01}.
Moreover, $F$ can be seen as a generator matrix of a double circulant code over
$\F_q$ \cite{MW}. Double circulant codes are a special case of quasi-cyclic
codes, which are a family of quadratic residue codes. Quasi-cyclic codes have
already been used for distributed storage \cite{Bl01}, which points out the
significance of these codes for DSS.

\begin{figure}
 \centering
 \begin{tikzpicture}[shorten >=1pt,->]
  \tikzstyle{source}=[rectangle,fill=black!25,minimum size=28pt,inner sep=0pt]
  \tikzstyle{cuteds}=[rectangle,fill=black!25,minimum size=26pt,inner sep=0pt]
  \tikzstyle{code}=[circle,fill=black!35,minimum size=50pt,inner sep=0pt]
  \tikzstyle{storage}=[rectangle,fill=none]

  \node[source] (SC) at (1,2.5) {Source};
  \node[cuteds] (a1) at (2.5,4) {$v_1$};
  \node[cuteds] (a2) at (2.5,3) {$v_2$};
  \node[cuteds] (a3) at (2.5,2) {$v_3$};
  \node[storage] (p1) at (2.5,1){$\vdots$};
  \node[cuteds] (a4) at (2.5,0) {$v_n$};

  \node[code] (MSR) at (4.7,2.5) {quasi-cyclic code};

  \node[cuteds] (c1) at (7,4.5) {$v_1$};
  \node[cuteds] (c2) at (7,3.7) {$\rho_1$};

  \node[cuteds] (c3) at (7,2.2) {$v_2$};
  \node[cuteds] (c4) at (7,1.5) {$\rho_2$};

  \node[storage] (p1) at (7,0.7){$\vdots$};

  \node[cuteds] (cn1) at (7,-0.3) {$v_n$};
  \node[cuteds] (cn) at (7,-1) {$\rho_n$};

  \node[storage] (n1) at (8,4.5) {$s_1$};
  \node[storage] (n2) at (8,2.2) {$s_2$};
  \node[storage] (n3) at (8,-0.3) {$s_n$};

  \node[storage] (exp1) at (2.5,5) {Information coordinates};
  \node[storage] (exp2) at (7,5.5) {Storage nodes};

  \foreach \from/\to in {SC/a1, SC/a2, SC/a3, SC/a4, a1/MSR, a2/MSR, a3/MSR,
a4/MSR, MSR/c1, MSR/c3, MSR/cn1}
    { \draw (\from) -- (\to); }
 \end{tikzpicture}
\caption{Construction process for a $[n,k,r]$ quasi-cyclic flexible MSR code.
}
\label{Fig:3}
\end{figure}
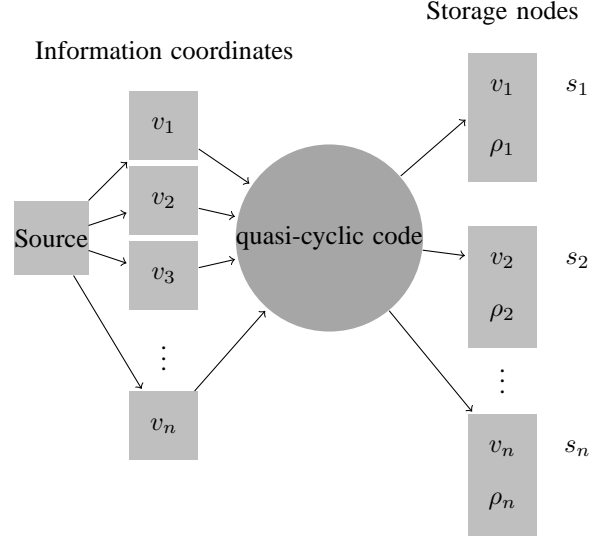

Figure \ref{Fig:3}, shows the construction of a QCFMSR code.
First, the file is split into $n$ symbols over $\F_q$. Then, these symbols are
encoded using $F$ and producing $2n$ symbols over $\F_{q}$. Finally, each two symbols are
stored together in one node, this creates the array code with coordinates over
$\F_q$ that can be seen as array coordinates over $\F_{q^2}$.

\subsection{Node Regeneration}
\label{Regeneration}

In this subsection, we show how to regenerate a failed node $s_i$, which stores
$(v_{i},\rho_i)$, minimizing the required repair bandwidth. Actually, we can
just follow the next algorithm:
\begin{enumerate}
 \item Download the information coordinates $v_j$, $j=i+1,\ldots,i+k \mod n$,
from the next $k$ nodes. Note that due to the circulant scheme, the
next node of $s_n$ is $s_1$. From these information coordinates, compute the
redundancy coordinate $\rho_i$ of the newcomer.

 \item Download the redundancy coordinate $\rho_{i-1}$ from the previous
node, following the same circulant scheme. Solving a simple equation, obtain the
information coordinate $v_i$ of the newcomer.
\end{enumerate}

It can be seen that $r=r_i=k+1$ for any $s_i$, $i=1,\ldots,n,$ and when the
repair problem is faced, it is clear that QCFMSR codes are optimal in terms of
the tradeoff curve given by the threshold function (\ref{eq:1}) for $r=k+1$.
Note that QCFMSR codes are in fact a family of regenerating codes because
$r>k$. However, unlike regenerating codes, for these flexible regenerating codes
the set of $r$ helping nodes is not any but an specific set of remaining
nodes with cardinality $r$. In other words, the set of nodes which is going to
send data to an specific newcomer is fixed.

Note that QCFMSR codes have also the exact repair property, which means that
once encoded, the information and the redundancy can be represented for the
whole life of the DSS by $c=(v_1,\ldots,v_n, \rho_1, \dots, \rho_n)$, where
$v_i$ and $\rho_i$ are the information and redundancy coordinates, respectively.
It is shown in \cite{Shah01} and \cite{Shah02} that when $r < 2k-3$,
exact MSR codes do not exist. However, QCFMSR codes exist for $r=k+1$ which
satisfies $r< 2k-3$ for $k>4$. These facts illustrate the importance of the
flexibility over the set of helper nodes in this construction.
Moreover, despite QCFMSR codes do not achieve uncoded repair, they are very
efficient regenerating one node, because they need only two simple operations
on the newcomer and no operation on the helper nodes.

\subsection{Data Reconstruction}
\label{DataReconstruction}

In Subsection \ref{CodeConstruction}, we have seen that $M=\alpha k$. In this
subsection, we prove that the array code over $\F_{q^2}$, used to construct a
QCFMSR code, satisfies that $d= n-k+1$ for some $\zeta_1,\ldots,\zeta_k$ and, as
a consequence, QCFMSR codes are MDS codes over $\F_{q^2}$ applied to DSS, so
they are MSR codes. In \cite{Ga01}, we performed a computational search to claim
the existence of QCFMSR codes. In this paper, we prove their existence
theoretically.

Let $F^s =(I^s|Z^s)$ denote the $n\times n$ submatrix of $F$
determined by $s=\{i_1,\ldots,i_k\}\in S$. Let $p_s(\zeta_1,\zeta_2, \ldots,
\zeta_k) \in \F_q[\zeta_1,\ldots,\zeta_k]$ be the multivariate polynomial
associated with the determinant of $F^s=(I^s|Z^s)$.

Assume that a DC wants to obtain the file. Then, it connects to any $k$ nodes
$\{s_{i_1},\ldots,s_{i_k}\}$ and downloads
$(v_{i_1},\rho_{i_1}),\ldots,(v_{i_k},\rho_{i_k})$, so the DC is downloading
the encoding given by $F^s$. In order to obtain the file given by $v = \left(
v_1,v_2,\ldots, v_{n} \right)$, we need $F^s$ to be full rank.
Moreover, in order to satisfy the data reconstruction condition, we need
$F^s$  to be full rank for all $s \in S$. Therefore, we have
to prove the following two statements:
\begin{enumerate}
 \item The polynomial associated with the determinant of $F^s$ is
not identically zero, which means that, when it is expanded as a summation of
terms, there exists at least one term with a nonzero coefficient.
 \item The polynomial $p_s(\zeta_1,\zeta_2, \ldots, \zeta_k)$ associated with
the determinant of $F^s$ is
nonzero with high probability, for a random choice of the nonzero coefficients
$\zeta_1,\ldots,\zeta_k$.
\end{enumerate}

We begin by proving the first statement. It is known that there exists a
relation between determinants of matrices and bipartite graphs.
Let $G(W_r \cup W_c,E)$ be the bipartite graph associated with a matrix $F^s$,
where each row of the matrix is represented by a vertex $w_{r_i}$ in $W_r$, and
each column of $F^s$ is represented by a vertex $w_{c_i}$ in $W_c$, where
$i=1,\dots,n$. Two vertices $w_{r_i}\in W_r$, $w_{c_i}\in W_c$ are adjacent if
the entry in the row $i$ and column $j$ of $F^s$ is nonzero. Moreover, the
weight of this edge is the nonzero value of this $i,j$th entry. Let
$E(w_{c_i})$ (resp. $E(w_{r_i})$) denote the neighbors of $w_{c_i}$
(resp. $w_{r_i}$) in the graph $G$. Let $T = \{ t_1,\ldots,t_m \} \subseteq W_c$
be a subset of vertices of $W_c$ or $T \subseteq W_r$ be a subset of vertices of
$W_r$ indistinctly. Let $E(T)$
denote the set $\bigcup_{i=1}^m E(t_i)$.

\begin{lemm}[\cite{Mot}]
\label{lema1}
The polynomial associated with the determinant of $F^s$, $p_s(\zeta_1,\zeta_2,
\ldots,\zeta_k)$, is not identically zero if and only if the bipartite graph
$G(W_r \cup W_c, E)$ associated with $F^s$ has a perfect matching.
\end{lemm}

To prove the existence of a perfect matching in the bipartite graph $G(W_r \cup
W_c,E)$ associated with $F^s$, we will use Hall's theorem.
\begin{lemm}[\cite{Hall}]
\label{hall}
A bipartite graph $G(W_r \cup W_c, E)$ contains a complete matching from $W_r$
to $W_c$ (resp. $W_c$ to $W_r$) if and only if it satisfies Hall's condition,
that is, for any $T \subseteq W_c$ (resp. $T \subseteq W_r$),
$|T| \le |E(T)|.$ Moreover, if $|W_r|=|W_c|$, the complete matching is
achieved in both directions, so it corresponds to a perfect matching.
\end{lemm}

\begin{lemm}
\label{neighboursUnionVC}
Let $T \subseteq W_c$ such that $T\not = \emptyset$, then $|E(T)| \ge |T|$.
\end{lemm}

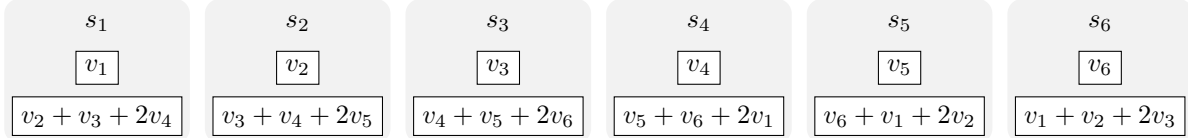
\begin{figure*}
 \begin{tikzpicture}

  \tikzstyle{nodess}=[rectangle,draw,fill=white]
  \tikzstyle{titol}=[rectangle,fill=gray!10]
\tikzstyle{background}=[rectangle, fill=gray!10, inner sep=0.1cm,rounded
corners=2mm]

  \matrix[row sep=0.15cm,column sep=0.2cm] {
       \node (v1) [titol]{$s_1$}; &
        &
       \node (v2) [titol]{$s_2$}; &
        &
       \node (v3) [titol]{$s_3$}; &
        &
       \node (v4) [titol]{$s_4$}; &
        &
       \node (v5) [titol]{$s_5$}; &
        &
       \node (v6) [titol]{$s_6$}; &
        \\

       \node (a1) [nodess]{$v_1$}; &
        &
       \node (a2) [nodess]{$v_2$}; &
        &
       \node (a3) [nodess]{$v_3$}; &
        &
       \node (a4) [nodess]{$v_4$}; &
        &
       \node (a5) [nodess]{$v_5$}; &
        &
       \node (a6) [nodess]{$v_6$}; &
        \\
       \node (r1) [nodess]{$v_2+v_3+2v_4$}; &
        &
       \node (r2) [nodess]{$v_3+v_4+2v_5$}; &
        &
       \node (r3) [nodess]{$v_4+v_5+2v_6$}; &
        &
       \node (r4) [nodess]{$v_5+v_6+2v_1$}; &
        &
       \node (r5) [nodess]{$v_6+v_1+2v_2$}; &
        &
       \node (r6) [nodess]{$v_1+v_2+2v_3$}; &
       \\
    };

    \begin{pgfonlayer}{background}
        \node [background,
                    fit=(v1) (r1)] {};
        \node [background,
                    fit=(v2) (r2)] {};
        \node [background,
                    fit=(v3) (r3)] {};
        \node [background,
                    fit=(v4) (r4)] {};
        \node [background,
                    fit=(v5) (r5)] {};
        \node [background,
                    fit=(v6) (r6)] {};
    \end{pgfonlayer}

 \end{tikzpicture}
\caption{
A $[6,3,4]$ QCFMSR code with coordinates over $\F_q$ and array coordinates over
$\F_q^2$.
}
\label{6_3_Scheme}
\end{figure*}

\begin{proof}
Note that $W_c$ has $k$ vertices of degree 1 and $k$ vertices of degree $k$.
We can decompose $T=T_1 \cup T_2$, where $T_1$ contains the vertices of degree 1
and $T_2$ contains the vertices of degree $k$.
It is clear that $|E(T_1)|=|T_1|$ by construction, and it is easy to see
that $|E(T_2)|\geq k+|T_2|-1 \geq |T_2|$ by the circular
construction of matrix $Z$ and because $k>1$. Therefore, we can assume that
$T_1\neq \emptyset$ and $T_2\neq \emptyset$.

If $|T_1|\leq k-1$, since $|E(T_1) \cap E(T_2)| \leq |E(T_1)| =
|T_1|$, we have that
$|E(T)| =|E(T_1)|+|E(T_2)|-|E(T_1) \cap E(T_2)| \geq
|T_1|+k+|T_2|-1-|T_1|\geq |T_1|+|T_2|=|T|$.
On the other hand, if $|T_1|=k$, then $|E(T_1) \cap E(T_2)| \leq
|T_1|-1$ since for each different vertex $t_i \in T_2$, there exists a
different vertex $t_{j} \in T_1$ such that $E(t_i) \cap E(t_{j}) =
\emptyset$. Thus, we also have
that $|E(T)|=|E(T_1)|+|E(T_2)|-|E(T_1)
\cap E(T_2)| \geq |T_1|+k+|T_2|-1-|T_1|+1 \geq |T_1|+|T_2|=|T|$.
\end{proof}

\begin{prop}
\label{indenticallyzero}
The polynomial associated with the determinant of $F^s$, $p_s(\zeta_1,\zeta_2,
\ldots,\zeta_k)$, is not identically zero.
\end{prop}
\begin{proof}
Since $|W_r|=|W_c|$, by using Lemmas \ref{hall} and \ref{neighboursUnionVC}, we
have that the bipartite graph $G(W_r\cup W_c,E)$ associated with
$F^s$ has a perfect matching. Finally, by Lemma \ref{lema1}, we know that
$p_s(\zeta_1,\zeta_2, \ldots,\zeta_k)$ is not identically zero.
\end{proof}

\medskip
For the second statement, we have to prove that for a random choice of the
nonzero coefficients $\zeta_1,\ldots,\zeta_n$, the multiplication of all the
multivariate polynomials associated with the determinant of all matrices $F^s,
s\in S$, is nonzero with high probability.

Let $p(\zeta_1,\ldots,\zeta_k) \in \F_q[\zeta_1,\ldots,\zeta_k]$ be the
multivariate polynomial $p(\zeta_1,\ldots,\zeta_k) = \prod_{s \in S}
p_s(\zeta_1,\ldots,\zeta_k)$. Note that if $p(\zeta_1,\ldots,\zeta_k) \ne 0$,
then $p_s(\zeta_1,\ldots,\zeta_k) \ne 0$ for all $s \in S$.

\begin{lemm}
\label{deg}
The degree of $p(\zeta_1,\ldots,\zeta_k)$ is less than or equal to $k {n
\choose k}$. Formally, $\deg(p(\zeta_1,\ldots,\zeta_k)) \le k {n \choose k}$.
\end{lemm}
\begin{proof}
Each $\zeta_i$, $i=1,\ldots,k$, can appear a maximum of $k$ times in
 $F^s$. By Lagrange minor's theorem, $p_s(\zeta_1,\ldots,\zeta_k)$ has a maximum
degree of $k$. By the definition of $p(\zeta_1,\ldots,\zeta_k)$,
$\deg(p(\zeta_1,\ldots,\zeta_k)) \le k {n \choose k}$.
\end{proof}

\begin{theo}
The ${n \choose k}$ submatrices $F^s, \; s \in S$, are full
rank with high probability for a sufficiently large finite field $\F_q$.
\end{theo}
\begin{proof}
By Proposition \ref{indenticallyzero} and using the Schwartz-Zippel lemma
\cite{Zi01}, we know that
$$\mbox{Pr}(p(\zeta_1,\ldots,\zeta_k) = 0) \le
\frac{\deg(p(\zeta_1,\ldots,\zeta_k))}{q}.$$
Therefore, $\mbox{Pr}(p(\zeta_1,\ldots,\zeta_k) \ne 0) \ge 1-
\frac{\deg(p(\zeta_1,\ldots,\zeta_k))}{q}$. And by Lemma \ref{deg}, we know
that $\deg(p(\zeta_1,\ldots,\zeta_k)) \le k {n \choose k}$, so for a
sufficiently large field
size $q$, submatrices $F^s, \; s \in S$, are full rank with high probability.
\end{proof}

\medskip
Summarizing, there is a set of full rank matrices $F^s$, $s \in S$, for a
random choice of the nonzero coefficients $\zeta_1,\ldots,\zeta_n$ and a
sufficiently large finite field. This means that there exists such $F$ that
represents a QCFMSR code with the property that any $k$ storage nodes have
enough information to reconstruct the file. In other words, a random choice of
the coefficients over a sufficiently large finite field gives the encoding for a
quasi-cyclic MDS array code of length $n$  over $\F_{q^2}$ where each array
coordinate of a codeword is $(v_i,\rho_i)$. As this code
is implemented in a DSS following the construction given in Subsection
\ref{CodeConstruction}, it gives a QCFMSR code.

It is worth to mention that using QCFMSR codes, an uncoded piece of the file is
always kept in the system. Moreover, if more than one storage node fails, up to
$n-k$, the decoding for the quasi-cyclic codes has linear complexity in contrast
with the one for Reed-Solomon codes, which has quadratic complexity \cite{MW}.

\subsection{Example}
\label{ExampleMSR}
In this subsection, we describe the construction of a $\left[ 6,3,4 \right] $
QCFMSR code over $\F_{5^2}$.

First, the file is fragmented into $6$ information coordinates $v =
\left( v_1,v_2,\ldots,v_{6} \right) $. Then, each $v_i$ for $i=1,\ldots, 6$
is stored in a node $s_i=(v_i,\rho_i)$, along with its corresponding
redundancy symbol $\rho_i$, which is computed using a quasi-cyclic
matrix $F$ of the following form:

\begin{equation}
\label{Example}
F = \left ( \begin{array}{cccccc|cccccc}
1&0&0&0&0&0 & 0 & 0 & 0 &\zeta_3&\zeta_2&\zeta_1\\
0&1&0&0&0&0 &\zeta_1& 0 & 0 & 0 &\zeta_3&\zeta_2\\
0&0&1&0&0&0 &\zeta_2&\zeta_1& 0 & 0 & 0 &\zeta_3\\
0&0&0&1&0&0 &\zeta_3&\zeta_2&\zeta_1& 0 & 0 & 0   \\
0&0&0&0&1&0 & 0 &\zeta_3&\zeta_2&\zeta_1& 0 & 0   \\
0&0&0&0&0&1 & 0 & 0 &\zeta_3&\zeta_2&\zeta_1& 0   \\
\end{array} \right ).
\end{equation}

By construction, the node regenerating condition is always achieved.
In order to satisfy the data reconstruction condition, we need to find
nonzero coefficients $\zeta_1,\zeta_2,\zeta_3$ such that
$p(\zeta_1,\zeta_2,\zeta_3)\ne 0$
over $\F_5$. Since $p(\zeta_1,\zeta_2,\zeta_3)=\zeta_1^{24} \zeta_2^{12}
\zeta_3^5 (-\zeta_2^2+\zeta_1 \zeta_3)^5
(\zeta_3^3+\zeta_1 ^3) (-\zeta_1\zeta_3^2+\zeta_2^2 \zeta_3)
(-\zeta_3^3-\zeta_1^3),$
a possible solution  over $\F_5$ is $(\zeta_1,\zeta_2,\zeta_3)=(1,1,2)$.
Figure \ref{6_3_Scheme} shows the distribution of
information and redundancy coordinates in the nodes. Each array
coordinate over $\F_{5^2}$ is represented by one storage node. It can be seen
that $d=4-2+1$ and that $\alpha k = M$, so the quasi-cyclic flexible code is a
MSR code.

Using the same argument, it is also possible to construct a
$\left[ 6,3,4 \right] $ QCFMSR code over $\F_{8^2}$ with nonzero coefficients
$(\zeta_1,\zeta_2,\zeta_3)=(1,1,z)$ over $\F_8$, where $z$ is a
primitive element of this field. Note that there is not any $[6,3,4]$ QCFMSR
code over $\F_{2^2}$, $\F_{3^2}$, $\F_{4^2}$ and $\F_{7^2}$.

\section{Regenerating codes with minimum bandwidth}
\label{MBRcodes}

In this section, we show how the QCFMSR codes are converted from minimum
storage codes to minimum bandwidth codes by using the same techniques used in
\cite{Ro01} and \cite{Ku03}, where MDS codes are used as base codes to construct
minimum bandwidth codes. Note that the minimum bandwidth is achieved with
the minimum $\alpha$ such that $\alpha=\gamma$. Moreover, we give some bounds
and conditions for the existence of these minimum bandwidth codes constructed using QCFMSR as base
codes. Specifically, in Subsections \ref{MBR_code_const}, \ref{MBR_node_regen} and
\ref{MBR_data_rec}, we show how to construct these codes in general and their
properties; and in Subsection \ref{MBR_nostres}, we focus on the quasi-cyclic case.

\subsection{Code Construction}
\label{MBR_code_const}

\begin{figure}
\centering
 \begin{tikzpicture}
  \tikzstyle{equa}=[rectangle,fill=black!25,minimum size=20pt,inner
sep=0pt,font=\tiny, text width=1.2cm, text centered]
  \tikzstyle{nodess}=[rectangle,fill=gray!20,minimum size=20pt,inner sep=2pt,
rounded corners=2mm]
  \tikzstyle{background}=[rectangle, fill=gray!10, inner sep=0.2cm,rounded
corners=5mm]
  \tikzstyle{titol}=[rectangle,fill=gray!10,font=\small]

  \node (MSR) [titol] at (0,4.2) {$[6,3,4]$ QCFMSR code};
  \node (MBR) [titol] at (0,-2.6) {$[6,2,2]$ quasi-cyclic flexible
regenerating code with
minimum bandwidth};

  \matrix[row sep=0.15cm,column sep=0.1cm] (M1) at (0,3) {
       \node (v1) [equa]{$v_1$}; &
	&
       \node (v2) [equa]{$v_2$}; &
	&
       \node (v3) [equa]{$v_3$}; &
	&
       \node (v4) [equa]{$v_4$}; &
	&
       \node (v5) [equa]{$v_5$}; &
	&
       \node (v6) [equa]{$v_6$}; &
       \\
       \node (v1') [equa]{$v_2+v_3+2v_4$}; &
       &
       \node (v2') [equa]{$v_3+v_4+2v_5$}; &
       &
       \node (v3') [equa]{$v_4+v_5+2v_6$}; &
       &
       \node (v4') [equa]{$v_5+v_6+2v_1$}; &
       &
       \node (v5') [equa]{$v_5+v_6+2v_1$}; &
       &
       \node (v6') [equa]{$v_1+v_2+2v_3$}; &
       \\
    };

  \matrix[row sep=0.15cm,column sep=0.1cm] (M2) at(0,-0.5) {

       \node (w1) [equa]{$v_1$}; &
	&
       \node (w2) [equa]{$v_2$}; &
	&
       \node (w3) [equa]{$v_3$}; &
	&
       \node (w4) [equa]{$v_4$}; &
	&
       \node (w5) [equa]{$v_5$}; &
	&
       \node (w6) [equa]{$v_6$}; &
       \\

       \node (w1') [equa]{$v_2+v_3+2v_4$}; &
       &
       \node (w2') [equa]{$v_3+v_4+2v_5$}; &
       &
       \node (w3') [equa]{$v_4+v_5+2v_6$}; &
       &
       \node (w4') [equa]{$v_5+v_6+2v_1$}; &
       &
       \node (w5') [equa]{$v_5+v_6+2v_1$}; &
       &
       \node (w6') [equa]{$v_1+v_2+2v_3$}; &
       \\

       \node (w22) [equa]{$v_2$}; &
	&
       \node (w23) [equa]{$v_3$}; &
	&
       \node (w24) [equa]{$v_4$}; &
	&
       \node (w25) [equa]{$v_5$}; &
	&
       \node (w26) [equa]{$v_6$}; &
	&
       \node (w21) [equa]{$v_1$}; &
       \\

       \node (w22') [equa]{$v_3+v_4+2v_5$}; &
       &
       \node (w23') [equa]{$v_4+v_5+2v_6$}; &
       &
       \node (w24') [equa]{$v_5+v_6+2v_1$}; &
       &
       \node (w25') [equa]{$v_5+v_6+2v_1$}; &
       &
       \node (w26') [equa]{$v_1+v_2+2v_3$}; &
	&
       \node (w21') [equa]{$v_2+v_3+2v_4$}; &
       \\
    };

    \begin{pgfonlayer}{background}
        \node [background,
                    fit= (MSR) (v1') (v6) (v6')] {};
        \node [background,
                    fit= (MBR) (w1') (w6) (w21')] {};
        \node (n1) [nodess,
                    fit=(v1) (v1')] {};
        \node (n2) [nodess,
                    fit=(v2) (v2')] {};
        \node (n3) [nodess,
                    fit=(v3) (v3')] {};
        \node (n4) [nodess,
                    fit=(v4) (v4')] {};
        \node (n5) [nodess,
                    fit=(v5) (v5')] {};
        \node (n6) [nodess,
                    fit=(v6) (v6')] {};

        \node (n11) [nodess,
                    fit=(w1) (w22')] {};
        \node (n12) [nodess,
                    fit=(w2) (w23')] {};
        \node (n13) [nodess,
                    fit=(w3) (w24')] {};
        \node (n14) [nodess,
                    fit=(w4) (w25')] {};
        \node (n15) [nodess,
                    fit=(w5) (w26')] {};
        \node (n16) [nodess,
                    fit=(w6) (w21')] {};

    \end{pgfonlayer}

  \foreach \from/\to in {n1.south/n11.north, n1.south/n16.north,
n2.south/n11.north, n2.south/n12.north, n3.south/n12.north, n3.south/n13.north,
n4.south/n13.north, n4.south/n14.north, n5.south/n14.north, n5.south/n15.north,
n6.south/n15.north, n6.south/n16.north}
    { \draw (\from) -- (\to); }

 \end{tikzpicture}

\caption{Construction of a $[6,2,2]$ quasi-cyclic flexible regenerating code
with minimum
bandwidth from a $[6,3,4]$ QCFMSR code.
}
\label{Example1}
\end{figure}
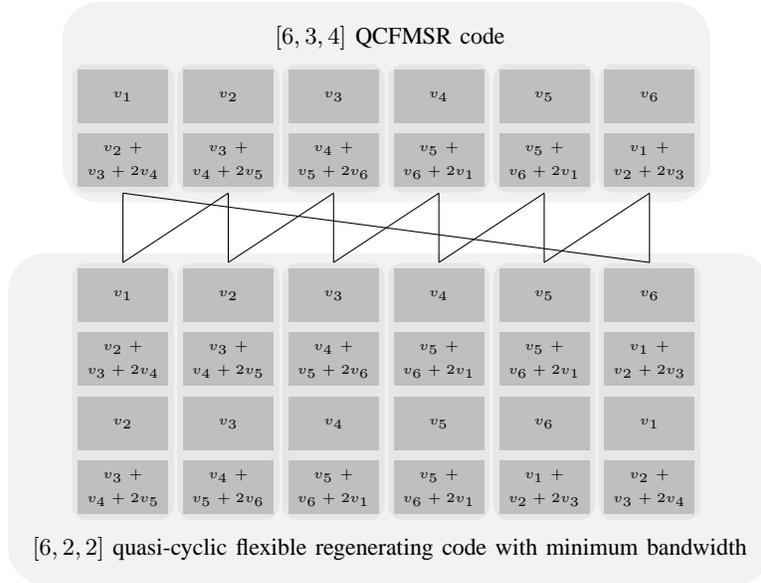

Let $C$ be a $[n,k,r]$ MSR code over $\F_{q^t}$, where any subset of
$k$ storage nodes is enough to reconstruct the file. In this
subsection, we explain how to construct a new $[\bar{n}, \bar{k}, \bar{r}]$
regenerating code $\bar{C}$ with minimum bandwidth over $\F_{q^{t\bar{r}}}$,
from the base code $C$.
Figure \ref{Example1} shows an example of a minimum bandwidth regenerating code
created from a $[6,3,4]$ QCFMSR code, which can be used to illustrate this
construction. Despite it is possible to construct $\bar{C}$ from any
regenerating code, the construction makes sense if the regenerating code $C$
is MSR, because then, $\bar{C}$ achieves the optimal parameters $\alpha$ and
$\gamma$.

\begin{lemm}
\label{lemm:1}
 Given $n \ge 3$, there exists a simple, undirected, and $\bar{r}$-regular graph
$H(W,E)$, where $W$ is the set of vertices where $|W| = \bar{n}$ and $E$ is the
set of edges where $|E| = n$, satisfying the following conditions: $|W|=
\bar{n}$, $|E|=n$,
$1<\bar{r}<\bar{n}$, and $\bar{r}\bar{n} = 2n$.
\end{lemm}
\begin{proof}
 Condition $\bar{r}\bar{n} = 2n$ is given by the Handshaking lemma for a simple,
undirected, and $\bar{r}$-regular graph. By Erdos-Gallai degree sequence
theorem, for $\bar{r} >1$ and $\bar{r} < \bar{n}$, there exists at least a
simple, undirected, and $\bar{r}$-regular graph such that $\bar{r}\bar{n} = 2n$.
\end{proof}

\medskip

For example, for $\bar{r} = \bar{n}-1$, we have the complete graph
$H=K_{\bar{n}}$,
and for $\bar{r}=2$ we have the cycle graph $H=C_{\bar{n}}$. As $|E| = n$
in $H(W,E)$, it is possible to assume that each edge in $E$ corresponds to a
different coordinate $c_j$ over $\F_{q^t}$ of a codeword
$c=(c_1, \ldots c_n) \in C$. Note that $c_j$ could be also seen as an array coordinate over
the base field $\F_q$, but in this section, we consider $c_j$ just as a coordinate over $\F_{q^t}$.
Given a codeword $c\in C$, since $|W|=\bar{n}$ in $H(W,E)$, we can construct a
codeword $\bar{c}=(\bar{c}_1,\ldots, \bar{c}_{\bar{n}}) \in \bar{C}$, where each
array coordinate ${\bar c}_i$ corresponds to a different vertex
$w_i \in W$ and contains the coordinates of $c$ given by the $\bar{r}$ edges
incident to $w_i$.
Moreover, since the graph is simple, any two vertices can not be connected by
more than one edge, so each coordinate of $c$ is contained in two array
coordinates of $\bar{c}$. As $C$ is defined over $\F_{q^t}$, $\bar{C}$ is
defined over $\F_{q^{t\bar{r}}}$.

In the next subsections, we prove that $\bar{C}$ is a regenerating code with minimum bandwidth.
Firstly, in Subsection \ref{MBR_node_regen} we show that $\gamma=\alpha$. Then,
in Subsection \ref{MBR_data_rec}, we look for the minimum $\alpha$ such that
any $\bar{k}$ array coordinates of $\bar{c} \in \bar{C}$ are enough to
reconstruct the file. Note that $\bar{C}$ is a regenerating code, but not a code
from the
classical point of view, since $|\bar{C}| = |C|$ and $\bar{k}$ is not the
dimension of the code but an integer such that $1 < \bar{k} < \bar{n}$.

\subsection{Node Regeneration}
\label{MBR_node_regen}
Assume that a storage node fails, which is the same as erasing one
array coordinate $\bar{c}_i$, $i=1,\ldots,\bar{n}$ of a codeword
$\bar{c}=(\bar{c}_1,\ldots,
\bar{c}_{\bar{n}}) \in \bar{C}$, or equivalently one vertex $\omega_i \in W$ of
$H(W,E)$. The newcomer can replace the failed node by downloading and storing
the $\bar{r}$ coordinates of $c$ included in each one of the $\bar{r}$
neighbors of $\omega_i$, and given by the corresponding $\bar{r}$ edges
incidents to $\omega_i$. According to this regeneration process,
$\gamma= \alpha$.

Note that these regenerating codes with minimum bandwidth, have the exact repair
and the uncoded repair properties. Also note that the node regeneration is given
by an specific subset of $\bar{r}$ helper nodes, so they can also be called
flexible regenerating codes.

\subsection{Data Reconstruction}
\label{MBR_data_rec}

Let $E(\omega_i)$, $i=1,\ldots,\bar{n}$, be the set of edges incident to the
vertex $\omega_i \in
W$. Let $\bar{S}$ be the set of all subsets of $\{1,\ldots,\bar{n}\}$ of
size $\bar{k}>1$ and let $\bar{s}\in\bar{S}$. For a subset $\bar{s}$,
$|\bigcup_{i\in\bar{s}} E(w_i)|=\sum_{i\in \bar{s}} |E(w_i)|-\theta_{\bar{s}}$,
where $\theta_{\bar{s}}$ represents the intersection terms in the
inclusion-exclusion formula. Since each edge is incident to two vertices, for
$\bar{s}'\subseteq \bar{s}$ of size $|\bar{s}'|>2$, $|\bigcap_{i\in\bar{s}'}
E(w_i)|=0$, so $\theta_{\bar{s}}=\sum_{i<j, \ i,j\in\bar{s}} |E(w_i)\cap
E(w_j)|$. Let $\theta$ be the maximum of all
$\theta_{\bar{s}}$, $\bar{s} \in \bar{S}$. Since $|E(w_i)\cap E(w_j)|\leq 1$ for any $i,j \in
\{1,\ldots,{\bar n}\}$ and $i\neq j$, we have that $\theta \leq {\bar{k}\choose
2}$.

\begin{lemm}
\label{theo:1}
Let $C$ be a $[n,k,r]$ MSR code over $\F_{q^t}$ with $n \ge 3$. Choose
$\bar{n}$, $\bar{k}$ and
$\bar{r}$ such that $\bar{r} \bar{n} = 2n$, $1< \bar{k}< \bar{n}$, $1< \bar{r} <
k$ and $k \le \bar{k} \bar{r} - \theta$. Then, there is a
$[\bar{n},\bar{k},\bar{r}]$ regenerating code $\bar{C}$ over
$\F_{q^{t\bar{r}}}$. Moreover, the minimum $\alpha$ is achieved when $k =
\bar{k}\bar{r}-\theta$.

\end{lemm}
\begin{proof}
Given a file distributed using $C$, we know that there are $n$ nodes and
that any $k$ of those $n$ nodes are enough to reconstruct the file. By Lemma
\ref{lemm:1}, we know that if $n \ge 3$, there exists a set of $\bar{n}$
vertices $W$ and a set of $n$ edges $E$, such that it is possible to construct
$H(W,E)$ with $1<\bar{r}<\bar{n}$ and $\bar{r}\bar{n} = 2n$. Then, from $H(W,E)$
it is possible to construct a code $\bar{C}$ as it is described in Subsection
\ref{MBR_code_const}.

The conditions $1 < \bar{k} < \bar{n}$ and $\bar{r} < k$ are necessary because
if they are not achieved, the code $\bar{C}$ has no sense as a regenerating
code. Note that a subset of cardinal $\bar{r}<k$ coordinates of $c$ contained in
$\bar{r}$ different nodes can regenerate a failed one, so $\gamma <M$. Finally,
in order to reconstruct the file distributed using $\bar{C}$, any subset of
$\bar{k}$ nodes must store at least $k$ coordinates of $c \in C$, so $k \le
|\bigcup_{i \in \bar{s}} E(\omega_i)|$. Since $k \leq \bar{k}\bar{r}-\theta \leq
\bar{k}\bar{r}-\theta_{\bar{s}} = |\bigcup_{i \in \bar{s}} E(\omega_i)|$ this
condition is satisfied. Therefore, $\bar{C}$ is a regenerating code.

In Subsection \ref{MBR_node_regen}, it is shown that $\gamma=\alpha$. Moreover,
$\bar{r}$ is the number of coordinates of $c$ which compose an array coordinate
of $\bar{c}$. Then, the minimum $\bar{r}$ will lead to the minimum $\alpha$.  As
$\bar{r} \ge (k+\theta)/\bar{k}$, the minimum $\bar{r}$ is achieved when $k =
\bar{k}\bar{r}-\theta$.
\end{proof}



As we are trying to minimize $\alpha$, we assume the equality $k =
\bar{k}\bar{r}-\theta$ given by Lemma \ref{theo:1}, and we establish  an upper
bound for the parameter $\theta$ in
Proposition \ref{k_IP}.
\begin{prop}
\label{k_IP}
In the graph $H(W,E)$ with $k = \bar{k}\bar{r}-\theta$, we
have that
\begin{enumerate}
 \item $\theta\le {\bar{k} \choose 2}$ if $\bar{k} \le \bar{r}+1$,
 \item $\theta\le \lfloor\frac{\bar{k}}{\bar{r}+1} \rfloor {\bar{r}+1 \choose
2}+
{\bar{k} \mod(\bar{r}+1) \choose 2}$ if $\bar{k} >\bar{r}+1$.
\end{enumerate}
\end{prop}
\begin{proof}
Each node $\omega_i \in W$ has $\bar{r}$ incident edges, so $\bar{k}$ nodes have
$\bar{k}\bar{r}$ edges. Now, we distinguish two cases.

 \textit{Case $\bar{k} \le \bar{r}+1$:} In $H(W,E)$, each vertex $\omega_i$,
$i=1,\ldots,\bar{n}$, shares one, and only one, edge with another vertex
$\omega_j$. Each vertex $\omega_{i}$, $i \in \bar{s}$ and $|\bar{s}| = \bar{k}$,
can share a maximum of one edge with each one of the other vertices $\omega_j, j
\in \bar{s}$, $i \ne j$.
Then, the maximum number of shared edges is ${\bar{k} \choose 2}$. In other
words, when $\bar{k} \le \bar{r}+1$, it is possible to create a complete
subgraph of $\bar{k}$ vertices in $H(W,E)$ with ${\bar{k} \choose 2}$ edges.

 \textit{Case $\bar{k} > \bar{r}+1$:} Given $H(W,E)$ and $\bar{s}$, we are going
to
construct a subgraph which maximizes the number of shared edges. Each vertex
$\omega_i$, $i \in \bar{s}$, can share a maximum of $\bar{r}$ edges with the
remaining
vertices $\omega_j, j \in \bar{s}$, $i \ne j$. Therefore, the maximum number of shared edges
is when we consider a complete
subgraph with $\bar{r}+1$ vertices and ${\bar{r}+1 \choose 2}$ edges.
As $\bar{k} > \bar{r}+1$, there could be
$\lfloor\frac{\bar{k}}{\bar{r}+1}\rfloor$
complete subgraphs, each one with ${\bar{r}+1 \choose 2}$ edges.
The vertices out of these complete subgraphs can share a maximum of
${\bar{k}
\mod (\bar{r}+1) \choose 2}$ edges, which leads to the upper bound $\theta \le
\lfloor\frac{\bar{k}}{\bar{r}+1} \rfloor {\bar{r}+1
\choose 2}+ {\bar{k} \mod(\bar{r}+1) \choose 2} $ for $\bar{k} > \bar{r}+1$.
\end{proof}

\subsection{Quasi-cyclic flexible regenerating codes with minimum bandwidth}
\label{MBR_nostres}

It is possible to use QCFMSR codes as base regenerating codes to
create regenerating codes with minimum bandwidth using the above construction.
In this subsection, we analyze the resulting parameters of these called
quasi-cyclic flexible regenerating codes with minimum bandwidth.

\begin{coro}
For $\bar{k} \le \bar{r}+1$, there exists a $[\bar{n}, \bar{k}, \bar{r}]$
quasi-cyclic flexible regenerating code with minimum bandwidth constructed from
a $[2k, k, k+1]$ QCFMSR code if the set of parameters
$\left\lbrace k, \bar{n}, \bar{k},\bar{r}\right\rbrace$ achieve:
$$ k = \frac{\bar{k}(2\bar{r}-\bar{k}+1)}{2},$$
$$\bar{n} = \frac{2\bar{k}(2\bar{r}-\bar{k}+1)}{\bar{r}},$$
$$ 1 < \bar{r} < k.$$
\end{coro}
\begin{proof}
Straightforward from Lemmas \ref{lemm:1}, \ref{theo:1} and Proposition
\ref{k_IP}.
\end{proof}

Figure \ref{Example1} shows an example of a quasi-cyclic flexible
regenerating code with
minimum bandwidth for $\bar{k} \le \bar{r}+1$ created from a $[6,3,4]$ QCFMSR
code. Each node $\omega_i\in W$ can be repaired downloading half node
$\omega_{i-1}$ and half node $\omega_{i+1}$. Moreover, any $\bar{k}=2$ nodes in
$\bar{C}$ contain at least $k=3$ different coordinates of $c \in C$, which allow
us to reconstruct the file. Note that $\alpha = \frac{2M}{3}$, which is equal
to the value given in \cite{Di01} for a $[6,3,4]$ MBR code.

\begin{figure*}
\centering
 \begin{tikzpicture}
  \tikzstyle{equa}=[rectangle,fill=black!25,minimum size=10pt,inner sep=0pt,
font=\tiny, text width=1.5cm, text centered]
  \tikzstyle{nodess}=[rectangle,fill=gray!20,minimum size=20pt,inner sep=2pt,
rounded corners=2mm]
  \tikzstyle{background}=[rectangle, fill=gray!10, inner sep=0.2cm,rounded
corners=5mm]
  \tikzstyle{titol}=[rectangle,fill=gray!10, font=\small]

  \node (MSR) [titol] at (0,3.5) {$[10,5,6]$ QCFMSR code};
  \node (MBR) [titol] at (0,-1.5) {$[10,4,2]$ quasi-cyclic felxible regenerating
code with minimum bandwidth};

  \matrix[row sep=0.15cm,column sep=0.05cm] (M1) at (0,2.5) {
       \node (v1) [equa]{$v_1$}; &
	&
       \node (v2) [equa]{$v_2$}; &
	&
       \node (v3) [equa]{$v_3$}; &
	&
       \node (v4) [equa]{$v_4$}; &
	&
       \node (v5) [equa]{$v_5$}; &
	&
       \node (v6) [equa]{$v_6$}; &
	&
       \node (v7) [equa]{$v_7$}; &
	&
       \node (v8) [equa]{$v_8$}; &
	&
       \node (v9) [equa]{$v_9$}; &
	&
       \node (v10) [equa]{$v_{10}$}; &
       \\
       \node (v1') [equa]{$v_2+5v_3+2v_4+v_5+v_6$}; &
       &
       \node (v2') [equa]{$v_3+5v_4+2v_5+v_6+v_7$}; &
       &
       \node (v3') [equa]{$v_4+5v_5+2v_6+v_7+v_8$}; &
       &
       \node (v4') [equa]{$v_5+5v_6+2v_7+v_8+v_9$}; &
       &
       \node (v5') [equa]{$v_6+5v_7+2v_8+v_9+v_{10}$}; &
       &
       \node (v6') [equa]{$v_7+5v_8+2v_9+v_{10}+v_1$}; &
       &
       \node (v7') [equa]{$v_8+5v_9+2v_{10}+v_1+v_2$}; &
       &
       \node (v8') [equa]{$v_9+5v_{10}+2v_1+v_2+v_3$}; &
       &
       \node (v9') [equa]{$v_{10}+5v_1+2v_2+v_3+v_4$}; &
       &
       \node (v10') [equa]{$v_1+5v_2+2v_3+v_4+v_5$}; &
       \\
    };

  \matrix[row sep=0.15cm,column sep=0.05cm] (M2) at (0,0) {

       \node (w1) [equa]{$v_1$}; &
	&
       \node (w2) [equa]{$v_2$}; &
	&
       \node (w3) [equa]{$v_3$}; &
	&
       \node (w4) [equa]{$v_4$}; &
	&
       \node (w5) [equa]{$v_5$}; &
	&
       \node (w6) [equa]{$v_6$}; &
	&
       \node (w7) [equa]{$v_7$}; &
	&
       \node (w8) [equa]{$v_8$}; &
	&
       \node (w9) [equa]{$v_9$}; &
	&
       \node (w10) [equa]{$v_{10}$}; &
       \\

       \node (w1') [equa]{$v_2+5v_3+2v_4+v_5+v_6$}; &
       &
       \node (w2') [equa]{$v_3+5v_4+2v_5+v_6+v_7$}; &
       &
       \node (w3') [equa]{$v_4+5v_5+2v_6+v_7+v_8$}; &
       &
       \node (w4') [equa]{$v_5+5v_6+2v_7+v_8+v_9$}; &
       &
       \node (w5') [equa]{$v_6+5v_7+2v_8+v_9+v_{10}$}; &
       &
       \node (w6') [equa]{$v_7+5v_8+2v_9+v_{10}+v_1$}; &
       &
       \node (w7') [equa]{$v_8+5v_9+2v_{10}+v_1+v_2$}; &
       &
       \node (w8') [equa]{$v_9+5v_{10}+2v_1+v_2+v_3$}; &
       &
       \node (w9') [equa]{$v_{10}+5v_1+2v_2+v_3+v_4$}; &
       &
       \node (w10') [equa]{$v_1+5v_2+2v_3+v_4+v_5$}; &
       \\

       \node (w22) [equa]{$v_2$}; &
	&
       \node (w23) [equa]{$v_3$}; &
	&
       \node (w24) [equa]{$v_4$}; &
	&
       \node (w25) [equa]{$v_5$}; &
	&
       \node (w26) [equa]{$v_6$}; &
	&
       \node (w27) [equa]{$v_7$}; &
	&
       \node (w28) [equa]{$v_8$}; &
	&
       \node (w29) [equa]{$v_9$}; &
	&
       \node (w210) [equa]{$v_{10}$}; &
	&
       \node (w21') [equa]{$v_1$}; &
       \\

       \node (w22') [equa]{$v_3+5v_4+2v_5+v_6+v_7$}; &
       &
       \node (w23') [equa]{$v_4+5v_5+2v_6+v_7+v_8$}; &
       &
       \node (w24') [equa]{$v_5+5v_6+2v_7+v_8+v_9$}; &
       &
       \node (w25') [equa]{$v_6+5v_7+2v_8+v_9+v_{10}$}; &
       &
       \node (w26') [equa]{$v_7+5v_8+2v_9+v_{10}+v_1$}; &
       &
       \node (w27') [equa]{$v_8+5v_9+2v_{10}+v_1+v_2$}; &
       &
       \node (w28') [equa]{$v_9+5v_{10}+2v_1+v_2+v_3$}; &
       &
       \node (w29') [equa]{$v_{10}+5v_1+2v_2+v_3+v_4$}; &
       &
       \node (w210') [equa]{$v_1+5v_2+2v_3+v_4+v_5$}; &
       &
       \node (w21') [equa]{$v_2+5v_3+2v_4+v_5+v_6$}; &
       \\
    };

    \begin{pgfonlayer}{background}
        \node [background,
                    fit= (MSR) (v1') (v10) (v10')] {};
        \node [background,
                    fit= (MBR) (w1') (w6) (w21')] {};
        \node (n1) [nodess,
                    fit=(v1) (v1')] {};
        \node (n2) [nodess,
                    fit=(v2) (v2')] {};
        \node (n3) [nodess,
                    fit=(v3) (v3')] {};
        \node (n4) [nodess,
                    fit=(v4) (v4')] {};
        \node (n5) [nodess,
                    fit=(v5) (v5')] {};
        \node (n6) [nodess,
                    fit=(v6) (v6')] {};
        \node (n7) [nodess,
                    fit=(v7) (v7')] {};
        \node (n8) [nodess,
                    fit=(v8) (v8')] {};
        \node (n9) [nodess,
                    fit=(v9) (v9')] {};
        \node (n10) [nodess,
                    fit=(v10) (v10')] {};

        \node (n11) [nodess,
                    fit=(w1) (w22')] {};
        \node (n12) [nodess,
                    fit=(w2) (w23')] {};
        \node (n13) [nodess,
                    fit=(w3) (w24')] {};
        \node (n14) [nodess,
                    fit=(w4) (w25')] {};
        \node (n15) [nodess,
                    fit=(w5) (w26')] {};
        \node (n16) [nodess,
                    fit=(w6) (w27')] {};
        \node (n17) [nodess,
                    fit=(w7) (w28')] {};
        \node (n18) [nodess,
                    fit=(w8) (w29')] {};
        \node (n19) [nodess,
                    fit=(w9) (w210')] {};
        \node (n20) [nodess,
                    fit=(w10) (w21')] {};
    \end{pgfonlayer}

  \foreach \from/\to in {n1.south/n11.north, n1.south/n20.north,
n2.south/n11.north, n2.south/n12.north, n3.south/n12.north, n3.south/n13.north,
n4.south/n13.north, n4.south/n14.north, n5.south/n14.north, n5.south/n15.north,
n6.south/n15.north, n6/n16.north, n7.south/n16.north, n7.south/n17.north,
n8.south/n17.north, n8.south/n18.north, n9.south/n18.north, n9.south/n19.north,
n10.south/n19.north,n10.south/n20.north}
    { \draw (\from) -- (\to); }

 \end{tikzpicture}

\caption{Construction of a $[10,4,2]$ quasi-cyclic flexible regenerating code
with minimum bandwidth from a $[10,5,6]$ QCFMSR code.
}
\label{Example2}
\end{figure*}
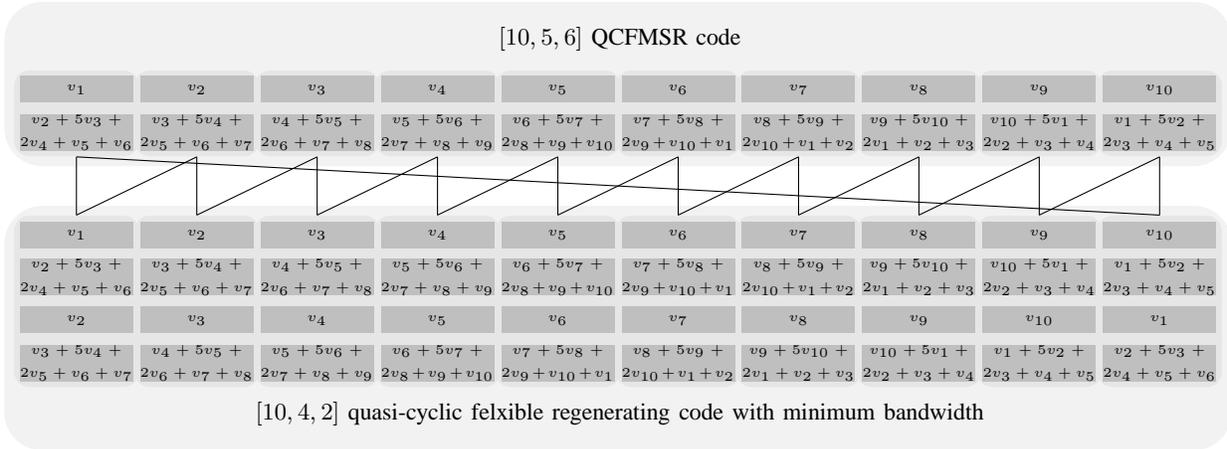

\begin{coro}
For $\bar{k} > \bar{r}+1$, there exists a $[\bar{n}, \bar{k}, \bar{r}]$
quasi-cyclic flexible regenerating code with minimum bandwidth constructed from
a $[2k, k, k+1]$ QCFMSR code if the set of parameters $\left\lbrace k, \bar{n},
\bar{k},\bar{r} \right\rbrace$ achieves:
$$ k = \bar{k}\bar{r} -\left \lfloor\frac{\bar{k}}{\bar{r}+1} \right \rfloor
{\bar{r}+1 \choose 2}+ {\bar{k}
\mod(\bar{r}+1) \choose 2},$$
$$\bar{n} = \frac{2n}{\bar{r}},$$
$$ 1 < \bar{r} < k .$$
\end{coro}
\begin{proof}
Straightforward from Lemmas \ref{lemm:1}, \ref{theo:1} and Proposition
\ref{k_IP}.
\end{proof}

Figure \ref{Example2} shows an example of a quasi-cyclic flexible
regenerating code with minimum bandwidth for $\bar{k} > \bar{r}+1$ created from
a $[10,5,6]$ QCFMSR code. Each node $\omega_i \in W$ can be repaired downloading
half node $\omega_{i-1}$ and half node $\omega_{i+1}$. Moreover, any $\bar{k}=4$
nodes in $\bar{C}$ contain at least $k=5$ different coordinates of $c \in C$,
which allows us to reconstruct the file. Note that $\alpha = \frac{2}{5} M$
which is less than $\frac{4}{7}M$, the value given in \cite{Di01} for a
$[10,5,6]$ MBR code. It is worth to mention that the reason of the decreasing on
the lower bound given in \cite{Di01} is the flexibility on the parameter
$\bar{r}$.

\begin{figure}
\begin{center}
\begin{tabular}{|cccc|}
\hline
$[\bar{n},\bar{k},\bar{r}]$ &
$[n,k,r]$ &
$\alpha=\gamma$ \cite{Di01} & $\alpha =
\gamma$\\
\hline
$[6,2,2]$ & $[6,3,4]$ &  $2M/3$ & $2M/3$ \\
$[8,3,3]$ & $[12,6,7]$ &  $M/2$ & $M/2$\\
$[7,2,4]$ & $[14,7,8]$ &   $4M/7$ & $4M/7$\\
$[10,4,4]$ & $[20,10,11]$ &  $2M/5$ & $2M/5$\\
\hline
$[10,4,2]$ & $[10,5,6]$ &  $4M/7$ & $2M/5$\\
$[10,5,2]$ & $[10,6,7]$ &  $5M/9$ & $M/3$\\
$[12,5,3]$ & $[18,9,10]$ &  $5M/12$ & $M/3$ \\
$[16,7,3]$ & $[24,12,13]$ &  $7M/18$ & $M/4$ \\
\hline
\end{tabular}
\end{center}
\caption{ \label{table:1} Parameters
$[\bar{n},\bar{k},\bar{r}]$ for some quasi-cyclic flexible regenerating codes
with minimum bandwidth constructed from $[n,k,r]$ QCFMSR codes, and comparison
between the $\alpha=\gamma$ values given in \cite{Di01} and the ones
achieved with the proposed construction.
}
\end{figure}
Figure \ref{table:1} shows the parameters of some quasi-cyclic flexible
regenerating codes with minimum bandwidth. The first column shows the
parameters $[\bar{n},\bar{k},\bar{r}]$ of the quasi-cyclic flexible regenerating
codes with minimum bandwidth. The second column shows the parameters $[n,k,r]$
of the corresponding QCFMSR codes.  The third and forth columns compare the
minimum $\alpha$ such that $\alpha=\gamma$ for MBR codes as stated in
\cite{Di01} with the one achieved by the quasi-cyclic flexible regenerating
codes with minimum bandwidth. First part of the table shows cases when $\bar{k}
\le \bar{r}+1$, and the second part cases when $\bar{k} > \bar{r}+1$

\section{Conclusions}
\label{Conclusions}

In this paper, we provide some interesting contributions to the current state
of the art of erasure codes applied to distributed storage. We construct a
family of regenerating codes using quasi-cyclic codes, where an
specific set of helper nodes is used to repair a storage node failure.
Firstly, we show a construction designed to minimize the storage per node, the
resulting codes are called quasi-cyclic flexible minimum storage regenerating
(QCFMSR) codes. Moreover, this construction is generalized to achieve bandwidth
optimality instead of storage optimality using graph techniques.

At the MSR point, we provide an exact repair solution for all parameters
achieving $r=k+1$ and $n=2k$. This construction is minimum according to the MSR
point in the fundamental tradeoff curve. Moreover, QCFMSR codes have
a very simple regenerating algorithm that approaches to the repair-by-transfer
property. In our solution, the helper nodes do not need to do any linear
combination among their symbols. The only linear combination is done in the
newcomer to obtain the symbols the first time that it enters into the system.
As far as we are concerned, this is the first construction achieving this
repair simplicity for the MSR point. We also claim that such codes exist with
high probability.

From an industrial point of view, it is interesting to have codes with high
rates, since these are the ones desired for actual data centers. Despite there
are constructions with equal \cite{Ta01} or higher \cite{Pa03} rate than QCFMSR
codes, their other properties (uncoded repair at the helper nodes, low decoding
and repairing complexity, good rate, low repair degree $r=k+1$ and exact repair)
makes them very interesting.

We also use a technique shown in \cite{Ro01} and \cite{Ku03} to construct codes
with minimum bandwidth using graphs and an existing MSR regenerating code. We
analyze and prove this construction giving bounds on the parameters of those
codes. This construction gives the minimum possible bandwidth $\gamma = \alpha$
achieved by an specific set of helper nodes and it has the repair-by-transfer
property. Finally, we show that QCFMSR codes can be used as base codes to
construct the quasi-cyclic flexible regenerating codes with minimum bandwidth.
We provide the conditions needed on the parameters
$\{n,k,d,\bar{n},\bar{k},\bar{r} \}$ for both cases, when $\bar{k} \le
\bar{r}+1$ and $\bar{k} > \bar{r}+1$.

\section*{Acknowledgment}

The authors would like to thank Professor Alexandros G. Dimakis for his helpful
and valuable comments and suggestions to this paper. This work has been
partially supported by the
Spanish MICINN grant TIN2010-17358, the Spanish Ministerio de Educaci\'on FPU
grant AP2009-4729 and the Catalan AGAUR grant 2009SGR1224.

\bibliographystyle{IEEEtran}
\bibliography{references}

\end{document}